\documentclass[journal]{IEEEtran}
\usepackage[utf8]{inputenc}

\usepackage{tablefootnote}
\usepackage{stfloats}

\usepackage[flushleft]{threeparttable}
\usepackage[pdftex]{graphicx}
\graphicspath{{./Graficas}}
\usepackage{float}
\usepackage{graphicx,amsfonts,latexsym}
\usepackage{dsfont}
\usepackage{subcaption}
\usepackage{verbatim}
\usepackage{amsmath} % assumes amsmath package installed
\usepackage{amssymb}  % assumes amsmath package installed
\usepackage{amsfonts}
\usepackage{amsthm} % For theorems, assumptions, remarks, etc

\newtheorem{definition}{Definition}
\newtheorem{lemma}{Lemma}

\newtheorem{remark}{Remark}

\newtheorem{corollary}{Corollary}

\newtheorem{proposition}{Proposition}

\usepackage{siunitx}

\usepackage[ruled, lined, ,longend, linesnumbered]{algorithm2e}
\usepackage[hang,flushmargin]{footmisc}
\usepackage{algpseudocode}
\usepackage{enumitem}
\SetKwInOut{Input}{Input}
\SetKwInOut{Declares}{Declares}
\SetKwInOut{Require}{Require}
\SetKwInput{Prototype}{Prototype}
\usepackage{algpseudocode}% <========== http://ctan.org/pkg/algorithmicx

\usepackage{xcolor}
\usepackage{hyperref}
\hypersetup{%
    colorlinks=true,
    linkcolor={blue!50!black},
    citecolor={blue!50!black},
    urlcolor={blue!50!black}
}

\usepackage{fancyhdr}

\newcommand {\bsis} {\left\{ \begin{array} }
\newcommand {\esis} {\end{array}\right.}

\def\bmat#1{\left[\begin{array}{#1}}
	\def\emat{\end{array}\right]}

\def\diag{\texttt{diag}}

\def\({\left(}
\def\){\right)}

\def\vv#1{{ \rm \bf{#1}}}
\def\Id{{\rm{I}}}

\def\R{\mathbb{R}}    % Set of real numbers

\newcommand {\T}{^{\top}} % transpose
\newcommand{\blista}{\renewcommand{\labelenumi}{(\roman{enumi})} % list with roman numbers.
	\begin{enumerate}}
	\newcommand{\elista}{\end{enumerate} \renewcommand{\labelenumi}{\arabic{enumi}.}}
\def\LBx{\underline{x}}
\def\UBx{\overline{x}}
\def\LBu{\underline{u}}
\def\UBu{\overline{u}}

\def\LBv{\underline{v}}
\def\UBv{\overline{v}}

\def\eLBx{\underline{x}_\varepsilon}
\def\eUBx{\overline{x}_\varepsilon}
\def\eLBu{\underline{u}_\varepsilon}
\def\eUBu{\overline{u}_\varepsilon}

\begin{document}

\title{\LARGE Efficient implementation of MPC for tracking using ADMM by decoupling its semi-banded structure} 

\author{Victor~Gracia$^{\dagger}$,~Pablo~Krupa$^\star$,~Daniel~Limon$^{\dagger}$,~Teodoro~Alamo$^{\dagger}$
    \thanks{$^\dagger$ Universidad de Sevilla, Seville, Spain.}
    \thanks{$^\star$ Gran Sasso Science Institute (GSSI), L'Aquila, Italy.}
    \thanks{This work was supported by grant PID2022-142946NA-I00 funded by MICIU/AEI/ 10.13039/501100011033 and by ERDF/EU.}
    \thanks{V. Gracia acknowledges support from grant PID2022-142946NA-I00 funded by MICIU/AEI/ 10.13039/501100011033 and by ESF+.}
    \thanks{P. Krupa acknowledges support from the MUR-PRO3 project on Software Quality and the MUR-PRIN project DREAM (20228FT78M).}
    \thanks{Corresponding author: Victor Gracia. E-mails: \texttt{vgracia@us.es}, \texttt{pablo.krupa@gssi.it}, \texttt{talamo@us.es}, \texttt{dlm@us.es}.}%
}

\maketitle
\pagestyle{fancy}
\thispagestyle{fancy}

\begin{abstract}
Model Predictive Control (MPC) for tracking formulation presents numerous advantages compared to standard MPC, such as a larger domain of attraction and recursive feasibility even when abrupt changes in the reference are produced. As a drawback, it includes some extra decision variables in its related optimization problem, leading to a semi-banded structure that differs from the banded structure encountered in standard MPC. This semi-banded structure prevents the direct use of the efficient algorithms available for banded problems.  To address this issue, we present an algorithm based on the alternating direction method of multipliers that explicitly takes advantage of the underlying semi-banded structure of the MPC for tracking.
\end{abstract}


\section{Introduction} \label{sec:introduction}

Model Predictive Control (MPC) is an advanced control policy whose control action is obtained from a constrained Optimization Problem (OP) posed at every sample time \cite{camacho2013model, rawlings2017model}.
MPC has become widely popular due to its ability to optimize the plant operation performance while dealing with constraints.
However, it presents disadvantages, such as the computational cost required to solve its associated OP at each sample time, or the unavailability of a suitable control action when the OP is infeasible or if it cannot be solved in a short amount of time compared to the sample time of the system.

Recently, there has been a significant amount of academic literature providing results which mitigate these issues, such as the proposal of efficient solvers suitable for the implementation of MPC \cite{Krupa_TCST_20, stellato2020osqp,odonoghue_SCS,Richter12,Patrinos14}, or results which palliate the problem of MPC infeasibility \cite{LIMON20082382, kohler2022globally}. 

In particular, in this article we focus on the \textit{MPC for tracking} (MPCT) formulation, originally proposed in~\cite{LIMON20082382}, which introduces an artificial reference as an additional decision variable of the OP.
The main benefit of introducing this artificial reference is that MPCT attains a notably larger domain of attraction and feasibility region when compared with standard MPC formulations.
Additionally, MPCT guarantees recursive feasibility, even when sudden changes in the reference occur.
Furthermore, it also guarantees asymptotic stability to an admissible steady state of the system, even if the reference is infeasible.
These benefits make MPCT a strong candidate for its use in a practical setting, since it significantly mitigates some of the main issues of standard MPC.
However, these benefits come at the cost of a more complex OP due to the addition of the artificial reference.
In particular, the inclusion of the additional decision variables leads to a semi-banded structure in the MPCT OP, whereas standard MPC formulations present a banded structure than can be exploited by the optimization solver \cite{Krupa_TCST_20, wang2009fast}. 

In \cite{krupa2021implementation}, the authors propose an efficient solver for MPCT where the banded structure of standard MPC formulations is recovered by using the Extended Alternating Direction Method of Multipliers (EADMM) \cite{cai2017convergence}.
This provides a solver whose computational cost per iteration is nearly identical to the one for standard MPC formulations using first-order optimization methods such as (non-extended) ADMM~\cite{MAL-016}.
However, the disadvantage is that EADMM presents several drawbacks when compared with ADMM, both theoretical and in terms of its practical performance.

In this article we present an alternative way of solving MPCT using the ADMM algorithm by decomposing the semi-banded structure of the MPCT OP.
This decomposition recovers the same banded structure associated with the standard MPC formulation, which can thus be exploited in the numerical solver.
The computational cost per iteration of the resulting solver is over two times larger than the one for the EADMM solver proposed in \cite{krupa2021implementation}.
However, in spite of this, the use of the ADMM algorithm instead of EADMM provides better theoretical guarantees and practical performance, as illustrated by the numerical results.

This article is structured as follows.
Section \ref{sec:mpct_formulation} introduces the MPCT formulation.
Section \ref{sec:ADMM_mpct_formulation} presents the proposed ADMM algorithm for MPCT, where we show how we decompose its semi-banded structure to attain an efficient solver.
Section \ref{sec:numerical_results} shows numerical results demonstrating the practical benefits of the proposed solver.
Finally, Section \ref{sec:conclusion} summarizes the main results of the article.

\subsubsection*{Notation}
Given a square matrix $A$, $\det(A)$ is its determinant.
Given two integers $a, b$, $\mathbb{I}_a^b = \{a, a + 1, \dots, b - 1, b \}$.
$\mathcal{S}_{\succ}^{n}$ denotes the set of symmetric positive definite matrices of size $n \times n$.
Given a vector $x$, we denote its $j$-th component as $x_{(j)}$.
Given $Q\in \mathcal{S}_{\succ}^{n}$,  $\| x \|_{Q} \doteq \sqrt{x\T Q x}$ and $\| x \|_\infty \doteq \max_{j = 1 \dots n}{| x_{(j)} |}$.
The identity matrix of dimension~$n$ is denoted by $\Id_{n}$ and the vector of ones by $\mathds{1}_n \in \R^n$ (we~may drop the sub-index $n$ if the dimension is clear from the context).
Given $x, y \in \R^n$, $x \leq (\geq) \; y$ denotes component-wise inequalities.
For vectors $x_1$ to $x_N$ of any dimension, $(x_{1}, \dots, x_{N})$ denotes the column vector formed by their concatenation.
We denote by $\diag(A_1, \dots, A_N)$ the block diagonal matrix formed by the concatenation of scalars and/or matrices $A_1$ to $A_N$ (possibly of different dimensions).
Function $\max(\cdot) \colon \R \times \R \times \dots \times \R \rightarrow \R$ returns the maximum element of its scalar inputs.
The Kronecker product between matrices $A$ and $B$ is denoted by $A \otimes B$.

\section{MPC for tracking formulation} \label{sec:mpct_formulation}

Consider a controllable discrete-time system described by
\begin{equation}\label{pred_model}
x(t+1) = Ax(t)+Bu(t),
\end{equation}
where $x(t) \in \R^{n_x}$ and $u(t) \in \R^{n_u}$ are the state and input at sample time $t$, respectively, subject to box constraints
\begin{subequations} \label{box_limits}
	\begin{align}
		\LBx & \leq  x(t) \leq \UBx,\\
		\LBu & \leq  u(t) \leq \UBu,
	\end{align}
\end{subequations}
where $\LBx, \UBx \in \R^{n_x}$ and $\LBu, \UBu \in \R^{n_u}$ satisfy $\LBx < \UBx$ and $\LBu< \UBu$. 
The control objective is to steer the system to the steady-state reference $(x_r,u_r)$ while satisfying the system constraints.
If $(x_r,u_r)$ is admissible, then the closed-loop system should converge to it.
Otherwise, we wish to converge to some admissible steady-state that is close to $(x_r,u_r)$.

This control objective can be achieved by using the \textit{MPC for tracking} (MPCT) proposed in \cite{LIMON20082382}. In particular, we address the MPCT formulation with terminal equality constraint from \cite{krupa2021implementation}, whose OP is given by
\begin{subequations}\label{MPCT_formulation}
	\begin{align}
		\min_{\substack{\vv{x,u}, \\x_s,u_s}} \; & \sum_{i=0}^{N-1}(\|x_{i}{-}x_{s}\|_{Q}^2 {+} \|u_{i}{-}u_{s}\|_{R}^2) {+} \|x_{s}{-}x_{r}\|_{T}^2 {+} \|u_{s}{-}u_{r}\|_{S}^2  \\ 
		\rm s.t. \; & x_{0} = x(t),  \label{initial_constraint}\\
		& x_{i+1} = Ax_{i} + Bu_{i}, \ i \in \mathbb{I}_{0}^{N-2},\\
		& x_{s} = Ax_{N-1} + Bu_{N-1},\\
        & x_{s} = Ax_{s} + Bu_{s},\\
        & \LBx \leq x_{i} \leq \UBx, \ i \in \mathbb{I}_1^{N-1}, \label{ineq_ini}\\
        & \LBu \leq u_{i} \leq \UBu, \ i \in \mathbb{I}_0^{N-1},\\
		& \eLBx \leq x_{s} \leq \eUBx, \label{ineq_pre_final}  \\
		& \eLBu \leq u_{s} \leq \eUBu, \label{ineq_final}
	\end{align}
\end{subequations}
where the decision variables are the artificial reference $(x_{s}, u_{s})$ and the predicted states $\vv{x} = (x_{0}, x_{1}, \dots, x_{N-1})$ and inputs $\vv{u} = (u_{0}, u_{1}, \dots, u_{N-1})$ along the prediction horizon $N$; $x(t)$ is the current state of the system at sample time $t$; the matrices $Q \in \mathcal{S}_{\succ}^{n_x}$, $R \in \mathcal{S}_{\succ}^{n_u}$, $T \in \mathcal{S}_{\succ}^{n_x}$ and $S \in \mathcal{S}_{\succ}^{n_u}$ are the cost function matrices; and given the arbitrarily small scalar $\varepsilon{>}0$, 
$\eLBx = \LBx+\varepsilon \mathds{1}_{n_x}$,
$\eUBx = \UBx-\varepsilon \mathds{1}_{n_x}$,
$\eLBu = \LBu+\varepsilon \mathds{1}_{n_u}$ and $\eUBu = \UBu-\varepsilon \mathds{1}_{n_u}$. The $\varepsilon$-tightened constraints \eqref{ineq_pre_final} and \eqref{ineq_final} are considered to avoid a possible controllability loss if any constraint is active at the equilibrium point \cite{LIMON20082382}.

The MPCT formulation \eqref{MPCT_formulation} has several advantages with respect to standard MPC \cite{camacho2013model}, such as guaranteed recursive feasibility under nominal conditions, i.e., when controlling the model used as prediction model with no disturbances, or asymptotic stability to the admissible steady state $(\hat{x}, \hat{u})$ that minimizes $\|\hat{x} - x_{r}\|_{T}^2 + \|\hat{u} - u_{r}\|_{S}^2$~\cite{LIMON20082382}.
However, the inclusion of $(x_s, u_s)$ leads to a more complex OP than the one of standard MPC, as the  banded structure that arises when solving the OP of  MPC is lost in \eqref{MPCT_formulation}. 
We note that the banded structure of standard MPC is crucial for the implementation of efficient solvers \cite{Krupa_TCST_20, wang2009fast}.
Thus, even though MPCT only adds $n_x + n_u$ extra decision variables with respect to standard MPC, the time required to solve \eqref{MPCT_formulation} can be notably higher if a naive approach is used to solve the OP, e.g., if non-sparse matrices are used when solving the OP.
In the following section we present an efficient ADMM-based solver for \eqref{MPCT_formulation}.

\section{Efficiently solving MPCT using ADMM} \label{sec:ADMM_mpct_formulation}

We now show how to efficiently solve \eqref{MPCT_formulation} using the ADMM algorithm \cite{MAL-016} by decomposing its most computationally expensive step into several simple-to-solve steps.
We start by describing the version of the ADMM we consider.

\subsection{Alternating Direction Method of Multipliers}

Let $f:\mathbb{R}^{n_z} \rightarrow (-\infty,\infty]$ and $g:\mathbb{R}^{n_z} \rightarrow (-\infty,\infty]$ be proper, closed and convex functions, $z$, $v \in \mathbb{R}^{n_z}$, and $C$, $D \in \mathbb{R}^{n_z \times n_z}$. Consider the OP
\begin{subequations}\label{ADMM_problem}
	\begin{align}
		\min_{z,v} & \quad f(z) + g(v) \\
		s.t. & \quad Cz+Dv=0,
	\end{align}
\end{subequations}
with augmented Lagrangian $\mathcal{L}_\rho : \mathbb{R}^{n_z} \times \mathbb{R}^{n_z} \times\mathbb{R}^{n_z} \rightarrow \mathbb{R}$,
\begin{equation*}
	\mathcal{L}_\rho(z,v,\lambda) = f(z) + g(v) + \lambda\T (Cz+Dv) + \frac{\rho}{2}\|Cz+Dv\|_2^2,
\end{equation*}
where $\lambda \in \mathbb{R}^{n_z}$ is the vector of dual variables and the scalar $\rho>0$ is the penalty parameter. We denote a solution of (\ref{ADMM_problem}) by ($z^*, v^*, \lambda^*$), provided that one exists.

Starting from an initial point $(v^0,\lambda^0)$, ADMM, shown in Algorithm~\ref{ADMM_algorithm}, returns a suboptimal solution $(\tilde{z}^*,\tilde{v}^*,\tilde{\lambda}^*)$ of \eqref{ADMM_problem}, where suboptimality is determined by the choice of the primal and dual exit tolerances $\epsilon_p>0$ and $\epsilon_{d}>0$ \cite[\S 3.3]{MAL-016}.

\begin{algorithm}[t]
	\DontPrintSemicolon
	\caption{ADMM} \label{ADMM_algorithm}
	\Require{$v^{0}$, $\lambda^{0}$, $\rho>0$, $\epsilon_{p}>0$, $\epsilon_{d}>0$}
	$k \gets 0$\;
	\Repeat{{$\|Cz^{k+1}{+}Dv^{k+1}\|_{\infty} {\leq} \epsilon_{p}$ and $\|v^{k+1}{-}v^{k}\|_{\infty} {\leq} \epsilon_{d}$}}{
		$z^{k+1} \gets \displaystyle \arg \min_{z} \mathcal{L}_{\rho}(z, v^{k}, \lambda^{k})$\; \label{step_z_ADMM}
		$v^{k+1} \gets \displaystyle \arg \min_{v} \mathcal{L}_{\rho}(z^{k+1}, v, \lambda^{k})$\; \label{step_v_ADMM}
		$\lambda^{k+1} \gets \lambda^{k} + \rho(Cz^{k+1} + Dv^{k+1})$\;
		$k \gets k+1$\;
	}		
	\KwOut{$\tilde{z}^{*} \gets z^{k}$, $\tilde{v}^{*} \gets v^{k}$, $\tilde{\lambda}^{*} \gets \lambda^{k}$}
\end{algorithm}

\subsection{Applying ADMM to MPCT}

For $y, \underline{y}, \overline{y} \in \R^{n_y}$, $\hat{G} \in \R^{m_y \times n_y}$, $\hat{b}\in \R^{m_y}$, let us define
\begin{equation*}
    \begin{aligned}
    \mathcal{I}_{[\underline{y},\overline{y}] }(y) &= 
\begin{cases}
	0,& \text{if} \ \underline{y} \leq y \leq \overline{y},\\
	+\infty,& \text{otherwise},
\end{cases} \\
    \mathcal{I}_{(\hat{G} y = \hat{b})}(y) &= 
\begin{cases}
	0,& \text{if} \ \hat{G} y = \hat{b},\\
	+\infty,& \text{otherwise}.
\end{cases}
\end{aligned}
\end{equation*}

Problem \eqref{MPCT_formulation} can be posed as \eqref{ADMM_problem} by taking $C = \Id$, $D=-\Id$,
\begin{equation*}
    z = (x_0, u_0, x_1, u_1, \dots, x_{N-1}, u_{N-1}, x_s, u_s),
\end{equation*}
and
\begin{equation*}
v = (\tilde{x}_0, \tilde{u}_0, \tilde{x}_1, \tilde{u}_1, \dots, \tilde{x}_{N-1}, \tilde{u}_{N-1}, \tilde{x}_s, \tilde{u}_s)
\end{equation*}
as a copy of the decision variables of \eqref{MPCT_formulation}, leading to
\begin{subequations}
\begin{align}
    f(z) &= \frac{1}{2} z\T H z + q\T z + \mathcal{I}_{(Gz=b)}(z), \label{eq:ADMM:MPCT:f} \\
    g(v) &= \mathcal{I}_{[\LBv,\UBv]}(v) = \mathcal{I}_{[\eLBx,\eUBx]} (\tilde{x}_s) + 
\mathcal{I}_{[\eLBu,\eUBu] } (\tilde{u}_s) \label{eq:ADMM:MPCT:g} \\
&\quad+ \sum_{i=1}^{N-1} \mathcal{I}_{[\LBx,\UBx]}( \tilde{x}_i) +  \sum_{i=0}^{N-1} \mathcal{I}_{[\LBu,\UBu]}( \tilde{u}_i), \nonumber
\end{align}
\end{subequations}
where $q = -(0, 0, \dots, 0, Tx_r, Su_r)$, $b = (x(t) , \dots, 0)$,
\begin{subequations} \label{eq:ADMM:MPCT:ingredients}
\begin{align}
		H &=\begin{bmatrix}
				Q & 0 & \cdots & -Q & 0\\
				0 & R & \cdots & 0 & -R\\
				0 & 0 & \ddots & \vdots & \vdots\\
				-Q & 0 & \cdots & NQ+T & 0\\
				0 & -R & \cdots & 0 & NR+S
			\end{bmatrix}, \label{eq:ADMM:MPCT:ingredients:H} \\
 		G &=\begin{bmatrix}
 				\Id & 0 & 0 & 0 & \cdots & 0\\
 				A & B & -\Id & 0 & \cdots &  0\\
 				0 & \ddots & \ddots & \ddots & 0 & \vdots\\
 				0 & 0 & A & B & -\Id & 0\\
 				0 & 0 & 0 & 0 & (A-\Id) & B
 			\end{bmatrix}, \label{eq:ADMM:MPCT:ingredients:G} \\
        \LBv &\doteq (\LBx, \LBu, \dots, \LBx, \LBu, \eLBx, \eLBu), \\
        \UBv &\doteq (\UBx, \UBu, \dots, \UBx, \UBu, \eUBx, \eUBu).
 \end{align}
\end{subequations}

\begin{comment}
	\begin{equation*}
		\underline{v} = (\underline{x},\underline{u}, \dots, \underline{x},\underline{u}, \underline{x}+\varepsilon_{x}, \underline{u}+\varepsilon_{u}),
	\end{equation*}
	\begin{equation*}
		\overline{v} = (\overline{x},\overline{u}, \dots, \overline{x}, \overline{u}, \overline{x}-\varepsilon_{x}, \overline{u}-\varepsilon_{u}),
	\end{equation*}
\end{comment}

\begin{comment}
	\[
	\mathcal{I}_{(\underline{v} \leq v \leq \overline{v})}(v)= 
	\begin{cases}
		0,& \text{if} \ \underline{v} \leq v \leq \overline{v},\\
		\infty,& \text{otherwise}.
	\end{cases}
	\]
\end{comment}
        
With these elements, Step~\ref{step_z_ADMM} of Algorithm~\ref{ADMM_algorithm} consists of solving a quadratic program subject to equality constraints, constituting the main computational load when solving MPCT with ADMM.
On the other hand, Step~\ref{step_v_ADMM} of the algorithm requires solving a simple separable convex problem (i.e., solving $n_z$ simple scalar OPs).
The next subsections are devoted to explaining how these steps are computed efficiently.

\subsection{Efficient computation of $z^{k+1}$}

Variable $z^{k+1}$ updated in Step~\ref{step_z_ADMM} of Algorithm~\ref{ADMM_algorithm} applied to problem \eqref{eq:ADMM:MPCT:f} is obtained from the optimal solution of
\begin{subequations} \label{z_problem}
    \begin{align}
\min_{z} & \; \frac{1}{2} z\T P z + p\T z \\
\rm  s.t. & \; Gz=b,
    \end{align}
\end{subequations}
where $P = H + \rho \Id$ and $p = q + \lambda^k - \rho v^k$.

As shown in the following proposition, problem \eqref{z_problem} can be solved by posing a linear system of equations describing its Karush-Kuhn-Tucker optimality conditions.

\begin{proposition}[{\cite[\S 5.5.3]{boyd2004convex}}]
Consider the OP~\eqref{z_problem}, where $P \in \mathbb{R}^{n_z \times n_z}$ is positive semi-definite, $p \in \mathbb{R}^{n_z}$, $G \in \mathbb{R}^{m_z \times n_z}$ and $ b \in \mathbb{R}^{m_z}$.
A vector $z^* \in \mathbb{R}^{n_z}$ is an optimal solution of this problem if and only if there exists a vector $\mu \in \mathbb{R}^{m_z}$ such that
\begin{subequations}\label{KKT_conditions_eqQP}
	\begin{align}
        &G z^* = b,\\
        &P z^* + G\T \mu + p = 0.
    \end{align}
 \end{subequations}
\end{proposition}

As shown in \cite{krupa2021tesis}, simple algebraic manipulations of \eqref{KKT_conditions_eqQP} along with the definition of matrix $W \doteq GP^{-1}G^\top$ lead to the alternative form
\begin{subequations} \label{KKT_conditions_eqQP_transformed}
\begin{align}
    &P \xi = p, \label{KKT_xi_computation}\\
    &W \mu = -(G \xi + b), \label{KKT_mu_computation}\\
    &P z^*= - (G\T \mu + p), \label{KKT_z_computation}
\end{align}
\end{subequations}
from where the optimal solution $z^*$ of \eqref{z_problem}, and thus the update $z^{k+1}$ of Step~\ref{step_z_ADMM} of Algorithm~\ref{ADMM_algorithm}, can be obtained.
Solving \eqref{KKT_conditions_eqQP_transformed} is the main computational burden of Algorithm~\ref{ADMM_algorithm}.
Thus, we wish to solve the three linear systems efficiently.
However, matrices $P$ and $W$ are semi-banded due to the semi-banded structure of $H$ shown in \eqref{eq:ADMM:MPCT:ingredients:H}.
The following definition formalizes the notion of a \textit{semi-banded} matrix.

\begin{definition} \label{def:semi}
Given the non-singular matrix $M\in \mathbb{R}^{n\times n}$ and vector $d \in \R^{n}$, we say that the linear system
\begin{equation} \label{eq:system:semi}
    M z = d,
\end{equation}
is {\it semi-banded}  if there exists a non-singular banded matrix $\Gamma \in \R^{n \times n}$, and $U \in \R^{n \times m}$ and $V \in \R^{m \times n}$ satisfying
\begin{equation} \label{eq:system:semi:decomposition}
    M = \Gamma + U V,
\end{equation}
where the dimension $m$ is assumed to be significantly smaller than the dimension of $M$, i.e., $m \ll n$.
\end{definition}

A naive approach to solving the three linear systems~\eqref{KKT_conditions_eqQP_transformed} will generally be computationally expensive.
However, we now show how the decomposition \eqref{eq:system:semi:decomposition} can be used to solve \eqref{KKT_conditions_eqQP_transformed} efficiently.
We start by showing that, indeed, matrices $P$ and $W$ in \eqref{KKT_conditions_eqQP_transformed} are semi-banded, providing explicit values for their decomposition \eqref{eq:system:semi:decomposition} in the following proposition, which makes use of the well-known Woodbury matrix identity \cite{1457855}.

\begin{lemma}[Woodbury matrix identity] \label{lemma:Woodbury}
Let $\Gamma \in \mathbb{R}^{n \times n}$ be non-singular. Then, if $\Id + V \Gamma^{-1} U$ is non-singular, $\Gamma + U V$ is also non-singular and its inverse is given by
\begin{equation} 
    (\Gamma + U V)^{-1} = \Gamma^{-1} - \Gamma^{-1} U ( \Id + V \Gamma^{-1} U)^{-1} V \Gamma^{-1}.\label{eq:Woodbury}
\end{equation}
\end{lemma}

\begin{proposition} \label{prop:semi-banded:values}
    Matrices $P$ and $W$ of \eqref{KKT_conditions_eqQP_transformed} are semi-banded and can be decoupled as $P = \widehat{\Gamma} + \widehat{U} \widehat{V}$ and $W = \tilde{\Gamma} + \tilde{U} \tilde{V}$, where, denoting $Y \doteq -\mathds{1}_{N}\T \otimes \diag(Q, R)$,
    \begin{align*}
        \widehat{\Gamma} &= \diag(Q, R, Q, R\dots, NQ+T, NR+S) + \rho \Id, \\
        \widehat{U} &= \begin{bmatrix}
				Y\T & 0 \\
				0 & \Id_{(n_x+n_u)}
			\end{bmatrix}, \quad
        \widehat{V} = \begin{bmatrix}
				0 & \Id_{(n_x+n_u)} \\
				Y & 0
			\end{bmatrix}, \\
        \tilde{\Gamma} &= G \widehat{\Gamma}^{-1} G\T,\\
        \tilde{U} &= -G \widehat{\Gamma}^{-1} \widehat{U} ( \Id {+} \widehat{V} \widehat{\Gamma}^{-1} \widehat{U} )^{-1}, \tilde{V} = \widehat{V} \widehat{\Gamma}^{-1} G\T.
    \end{align*}
Moreover, provided that $G$ is full column rank, matrices $P$, $W$,  $\widehat{\Gamma}$ and $\tilde{\Gamma}$ are positive definite.  
\end{proposition}

\begin{proof}
The decomposition $P = \widehat{\Gamma} + \widehat{U} \widehat{V}$ immediately follows from the definition of $P \doteq H + \rho \Id$ and the semi-banded structure of $H$ shown in \eqref{eq:ADMM:MPCT:ingredients:H}.
Next, by applying \eqref{eq:Woodbury} to $P$ we have that
\begin{equation*}
    P^{-1} = \widehat{\Gamma}^{-1} - \widehat{\Gamma}^{-1} \widehat{U} (\Id + \widehat{V} \widehat{\Gamma}^{-1} \widehat{U})^{-1} \widehat{V} \widehat{\Gamma}^{-1}.
\end{equation*}
Thus, from the definition $W \doteq G P^{-1} G\T$, we have
\begin{align*}
    W &= G ( \widehat{\Gamma}^{-1} - \widehat{\Gamma}^{-1} \widehat{U} (\Id + \widehat{V} \widehat{\Gamma}^{-1} \widehat{U})^{-1} \widehat{V} \widehat{\Gamma}^{-1} ) G\T \\
    &= G \widehat{\Gamma}^{-1} G\T - G \widehat{\Gamma}^{-1} \widehat{U} (\Id + \widehat{V} \widehat{\Gamma}^{-1} \widehat{U})^{-1} \widehat{V} \widehat{\Gamma}^{-1} G\T,
\end{align*}
from where the claim $W = \tilde{\Gamma} + \tilde{U} \tilde{V}$ then follows from the definitions of $\tilde{\Gamma}$, $\tilde{U}$ and $\tilde{V}$.
Finally, the fact that $\tilde{\Gamma}$ is banded-diagonal follows from the banded-diagonal structures of $\widehat{\Gamma}$ and $G$ \eqref{eq:ADMM:MPCT:ingredients:G}, as shown in \cite[Eq. (33)]{Krupa_TCST_20}.

We notice that $H$ is the matrix that corresponds to the quadratic terms of MPC formulation \eqref{MPCT_formulation}. From the convexity of the quadratic cost codified by $H$, we infer that $H\succeq 0$. Therefore, $P=H+\rho \Id \succeq \rho \Id\succ 0$. Also, from the positive definite nature of matrices $Q$, $R$, $S$ and $T$, we have that the block-diagonal matrix $\widehat{\Gamma}$ satisfies $\widehat{\Gamma}\succeq \rho \Id \succ 0$. Since $P$ is positive definite, $W=GP^{-1} G\T$ is also positive definite provided that $G$ is full column rank. The same argument applies to  $\tilde{\Gamma} = G \widehat{\Gamma}^{-1} G\T$.
\end{proof}

Next, we show in the following proposition that \eqref{eq:system:semi}, and thus \eqref{KKT_conditions_eqQP_transformed} by virtue of Proposition~\ref{prop:semi-banded:values}, can be solved by means of Algorithm~\ref{semibanded_algorithm}.
This result also follows from Lemma~\ref{lemma:Woodbury}.

\begin{proposition} \label{prop:solve:system:semi}
Consider the semi-banded system $M z = d$ of Definition~\ref{def:semi} and its decomposition $M = \Gamma + U V$, where both $M$ and $\Gamma$ are non-singular matrices.
Algorithm~\ref{semibanded_algorithm} returns a solution $\tilde{z}$ satisfying $M \tilde{z} = d$.
\end{proposition}

\begin{proof}
From $\det(M)\neq 0$ and $\det(\Gamma)\neq 0$ we obtain
\begin{eqnarray}
0& \neq & \det(\Gamma+UV) = \det(\Gamma)\det(\Id_n+\Gamma^{-1} UV) \nonumber\\
& = & \det(\Gamma)\det(\Id_m+V\Gamma^{-1}U), \label{equ:det:not:zero}
\end{eqnarray}
where the last equality is due to the well-known identity $\det(\Id_n+AB)=\det(\Id_m+BA)$, $\forall A\in\mathbb{R}^{n\times m}, \forall B\in \mathbb{R}^{m\times n}$. 
Thus, we infer from \eqref{equ:det:not:zero} that $\Id_m + V \Gamma^{-1} U$ is non-singular. From this, and \eqref{eq:Woodbury}, we have that $M^{-1} = (\Gamma + U  V)^{-1}$ can be written as 
$M^{-1} = \Gamma^{-1} - \Gamma^{-1} U (\Id_m + V \Gamma^{-1} U)^{-1} V \Gamma^{-1}$.
Therefore,
\begin{equation} \label{x_definition}
    \tilde{z} = M^{-1} d =  \Gamma^{-1} d - \Gamma^{-1} U (\Id_m+ V \Gamma^{-1} U)^{-1} V \Gamma^{-1} d.
\end{equation}
Defining $z_1 = \Gamma^{-1} d$, $z_2 = (\Id_m + V \Gamma^{-1} U)^{-1} V \Gamma^{-1} d$ and $z_3 = \Gamma^{-1} U (\Id_m + V \Gamma^{-1} U)^{-1} V \Gamma^{-1} d$, we obtain Step~\ref{SB_alg_step_1} of Algorithm~\ref{semibanded_algorithm} by the definition of $z_1$, Step \ref{SB_alg_step_2} by including the definition of $z_1$ into $z_2$, and Step \ref{SB_alg_step_3} by substituting the definition of $z_2$ into $z_3$.
Finally, $\tilde{z} = z_1 - z_3$ by substituting the definitions of $z_1$ and $z_3$ into \eqref{x_definition}.%\qedhere
\end{proof}

We notice that the computation of $z_1$ and $z_3$ in Algorithm~\ref{semibanded_algorithm} can be done efficiently by exploiting the banded structure of $\Gamma$, e.g., using a banded Cholesky decomposition if $\Gamma$ is positive definite \cite{Krupa_TCST_20}.
Moreover, since $\Id_m + V \Gamma^{-1} U\in \mathbb{R}^{m\times m}$, and we assume that $m \ll n$, $z_2$ is the solution of a small-dimensional linear system, which is thus computationally cheap to solve in comparison to Steps~\ref{SB_alg_step_1} and \ref{SB_alg_step_3} of Algorithm~\ref{semibanded_algorithm}.

\begin{algorithm}[t]
	\DontPrintSemicolon
	\caption{Solve semi-banded system  $(\Gamma+UV)\tilde{z}{=}d$} \label{semibanded_algorithm}
	\Require{$\Gamma$, $U$, $V$, $d$}
    Compute $z_1$ solving $\Gamma z_1 = d$\; \label{SB_alg_step_1}
    Compute $z_2$ solving $(\Id + V \Gamma^{-1} U) z_2 = V z_1$\; \label{SB_alg_step_2}
    Compute $z_3$ solving $\Gamma z_3 = U z_2$\; \label{SB_alg_step_3}
    \KwOut{$\tilde{z} \gets z_1-z_3$}
\end{algorithm}

\begin{corollary}
By means of Proposition~\ref{prop:solve:system:semi}, the optimal solution $z^*$ of problem \eqref{z_problem} can be obtained by using Algorithm~\ref{semibanded_algorithm} to solve the three linear systems \eqref{KKT_conditions_eqQP_transformed} using the decomposition of matrices $P$ and $W$ provided in Proposition~\ref{prop:semi-banded:values}.
\end{corollary}

\begin{remark} \label{rem:structures}
Matrix $\widehat{\Gamma}\succ 0$ in Proposition~\ref{prop:semi-banded:values} is block-diagonal.
Therefore, Steps~\ref{SB_alg_step_1} and \ref{SB_alg_step_3} of Algorithm~\ref{semibanded_algorithm} applied to solve \eqref{KKT_xi_computation} and \eqref{KKT_z_computation} are very simple.
On the other hand, matrix $\tilde{\Gamma}\succ 0$ is banded-diagonal, but not block-diagonal.
However, $\tilde{\Gamma}$ has the same banded-diagonal structure that is exploited by the solvers proposed in \cite{Krupa_TCST_20, krupa2021implementation}.
Therefore, Steps~\ref{SB_alg_step_1} and \ref{SB_alg_step_3} of Algorithm~\ref{semibanded_algorithm} applied to \eqref{KKT_mu_computation} can be solved by computing the banded Cholesky decomposition of $\tilde{\Gamma}$ and using \cite[Alg. 11]{krupa2021tesis}.
\end{remark}

\begin{remark}
Note that the operations $G \xi$ and $G\T \mu$ in \eqref{KKT_mu_computation} and \eqref{KKT_z_computation} can be performed sparsely, since $G$ \eqref{eq:ADMM:MPCT:ingredients:G} is sparse.
\end{remark}

\subsection{Computation of $v^{k+1}$}

Variable $v^{k+1}$ updated in Step~\ref{step_v_ADMM} of Algorithm~\ref{ADMM_algorithm}, when $g(v)$ is given by \eqref{eq:ADMM:MPCT:g}, is taken from the optimal solution of
\begin{equation} \label{QP_ineq_constrained}
    \begin{aligned}
        \min_{v \in \R^{n_z}} &\frac{\rho}{2} \sum_{j=1}^{n_z} (v_{(j)}^2 - 2z^{k+1}_{(j)}v_{(j)}) - \sum_{j=1}^{n_z} \lambda^k_{(j)} v_{(j)} \\
  \textrm{s.t.} &\; \LBx \leq x_{i} \leq \UBx, \ i \in \mathbb{I}_1^{N-1}, \\
        & \;  \LBu \leq u_{i} \leq \UBu, \ i \in \mathbb{I}_0^{N-1},\\
		& \; \eLBx \leq x_{s} \leq \eUBx,\\
		& \; \eLBu \leq u_{s} \leq \eUBu,
    \end{aligned}
\end{equation}
which is separable for each decision variable $v_{(j)}$.
Indeed, each element $v_{(j)}^{k+1}$, $j \in \mathbb{I}_1^{n_z}$, of $v^{k+1}$ is given by
\begin{equation} \label{eq:update_v}
    v_{(j)}^{k+1} = \min \left( \max \left( z_{(j)}^{k+1} + \frac{1}{\rho} \lambda_{(j)}^k,\underline{v}_{(j)} \right),\overline{v}_{(j)} \right).
\end{equation}

\subsection{Comparison with the EADMM-based solver}

\begin{algorithm}[t]
	\DontPrintSemicolon
	\caption{Efficient ADMM applied to MPCT \eqref{MPCT_formulation}} \label{alg:ADMM_for_MPCT}
	\Require{$\rho > 0$, $\epsilon_{p} > 0$, $\epsilon_{d} > 0$, $Q$, $R$, $S$, $T$, $N$}
    \Input{$x(t)$, $x_r$, $u_r$, $v^{0}$, $\lambda^{0}$,}
	$k \gets 0$\;
    Compute $q$ and $b$ in \eqref{eq:ADMM:MPCT:f} using $x(t)$, $x_r$ and $u_r$.\;
	\Repeat{$\|z^{k+1} - v^{k+1}\|_{\infty} \leq \epsilon_{p}$ and $\|v^{k+1} -v^{k}\|_{\infty} \leq \epsilon_{d}$}{
        $p \gets q + \lambda^k - \rho v^k$\;
        $\xi \gets$ solution of $P \xi = p$ using Alg.~\ref{semibanded_algorithm} \label{alg:ADMM_for_MPCT:xi}\;
        $\mu \gets$ solution of $W \mu = -(G \xi + b)$ using Alg.~\ref{semibanded_algorithm} \label{alg:ADMM_for_MPCT:mu}\;
		$z^{k+1} \gets$ solution of $P z {=} -(G\T \mu + p)$ using Alg.~\ref{semibanded_algorithm}\label{alg:ADMM_for_MPCT:z}\;
        Update $v^{k+1}$ using \eqref{eq:update_v} \label{alg:ADMM_for_MPCT:v}\;
        $\lambda^{k+1} \gets \lambda^{k} + \rho(z^{k+1} - v^{k+1})$\;
		$k \gets k+1$\;
	}		
	\KwOut{$u(t) \gets$ elements of $v^k$ corresponding to $\tilde{u}_0$}
\end{algorithm}

Algorithm~\ref{alg:ADMM_for_MPCT} shows the particularization of Algorithm~\ref{ADMM_algorithm} applied to the MPCT problem \eqref{MPCT_formulation} using the results presented in the previous subsections.
Steps~\ref{alg:ADMM_for_MPCT:xi}, \ref{alg:ADMM_for_MPCT:mu} and \ref{alg:ADMM_for_MPCT:z} of Algorithm~\ref{alg:ADMM_for_MPCT} are its main computational burden.
They make use of Algorithm~\ref{semibanded_algorithm} to solve the three semi-banded linear systems in \eqref{KKT_conditions_eqQP_transformed} to update the decision variables $z^{k+1}$.
As discussed in Remark~\ref{rem:structures}, Step~\ref{alg:ADMM_for_MPCT:mu} requires solving two linear systems whose matrix has the structure obtained when solving standard MPC formulations.
The EADMM solver \eqref{MPCT_formulation} proposed in \cite{krupa2021implementation} also recovers this very same banded structure.
However, it only needs to solve it once, instead of twice.
Additionally, the EADMM solver needs a problem computationally equivalent to Step~\ref{alg:ADMM_for_MPCT:v} of Algorithm~\ref{alg:ADMM_for_MPCT} and a $(n_x + n_u)$ dense linear system, whereas Algorithm~\ref{alg:ADMM_for_MPCT} needs to solve three 
$2(n_x + n_u)$-dimensional linear system (one for each call to Algorithm~\ref{semibanded_algorithm}) and four block-diagonal systems $\hat{\Gamma} z = d$ (two in Step~\ref{alg:ADMM_for_MPCT:xi} and two in Step~\ref{alg:ADMM_for_MPCT:z}).

We conclude that the computational cost per iteration of Algorithm~\ref{alg:ADMM_for_MPCT} is more than double the one of the EADMM solver proposed in \cite{krupa2021implementation}.
Therefore, ADMM should require, on average, less than half the number of iterations than the EADMM solver from \cite{krupa2021implementation} to be computationally better.
However, the are several advantages to using ADMM instead of EADMM.
First, the convergence of EADMM is only guaranteed if its step-size $\rho$ belongs to a certain range that depends on the properties of the OP \cite[Theorem~3.1]{cai2017convergence}.
However, in practice the EADMM algorithm typically performs very poorly when using  values of $\rho$ satisfying this theoretical condition.
On the other hand, the value of $\rho$ for Algorithm~\ref{alg:ADMM_for_MPCT} can be freely chosen, thus improving the practical performance of the algorithm.
Second, the theoretical convergence results for ADMM are better than the current ones available for EADMM, leading also to a better worst-case iteration complexity.
Finally, even when choosing values of $\rho$ for EADMM according to \cite[Theorem~3.1]{cai2017convergence}, the number of iterations required by the algorithm in practice is typically much larger than the ones required by ADMM when applied to the same OP.

\section{Numerical results} \label{sec:numerical_results}

We provide a computational comparison between the proposed Algorithm~\ref{alg:ADMM_for_MPCT} and the EADMM solver for MPCT presented in \cite{krupa2021implementation}.
We consider the ball and plate system presented in \cite[\S V.A]{9309006}, which consists of a ball whose position on a (nominally) horizontal plate is controlled by motors on each of its two main axes.
Consequently, the system has two inputs, angular accelerations [rad/s$^2$] of the motors, and eight states, position [m] and velocity [m/s] of the ball with respect to each axis, as well as angular position [rad] and velocity [rad/s] of the plate in each axis.
The physical parameters of the system are the same as in \cite[\S V.A]{9309006}, as well as the sample time of $0.2$ seconds.
We take the constraints \eqref{box_limits} as $\UBu = (0.2,0.2)$, $\LBu = -\overline{u}$, $\UBx = (2,1,0.785,\infty,2,1,0.785,\infty)$, $\LBx = -(0,1,0.785,\infty,0,1,0.785,\infty)$, and, for \eqref{MPCT_formulation}, $N {=} 30$, $\varepsilon {=} 10^{-6}$, $R = \diag(0.5,0.5)$, $S = \diag(0.3,0.3)$,
\begin{align*}
    Q &= \diag(10,0.05,0.05,0.05,10,0.05,0.05,0.05), \\
    T &= \diag(200,50,50,50,200,50,50,50),
\end{align*}
where the order of the state and input elements are taken from~\cite{9309006}.
Additionally, we take the exit tolerances of Algorithm~\ref{alg:ADMM_for_MPCT} and EADMM as $10^{-4}$, i.e., $\epsilon_p = \epsilon_d = 10^{-4}$.
Finally, we note that we use the scaling of the system model used in \cite{krupa2021implementation} to improve the numerical conditioning of \eqref{MPCT_formulation}.

{\renewcommand{\arraystretch}{0.72}%
\begin{table*}[t]
    \centering
    \footnotesize
    \begin{tabular}{ccc|cccc|cccc}
        & & & \multicolumn{4}{c|}{Iterations} &   \multicolumn{4}{c}{Computation time [ms]}  \\ \hline
        Algorithm & $\rho$ & Reference & Average         & Median        & Max.    & Min. & Avg. & Median & Max. & Min.
        \\ \hline 
        ADMM     & 0.1        & Reachable  &  14.2        &  7.0           & 	73.0           &  5.0 & 0.52 & 0.24 & 4.14 & 0.18\\
        ADMM     & 0.6        & Reachable  &  16.5        &  17.0           & 	18.0           &  11.0 & 0.60 & 0.60 & 0.12 & 0.39\\
        ADMM     & 2        & Reachable  & 37.8        &  39.0           & 	43.0               & 22.0 & 1.36 & 1.37 & 2.55 & 0.78\\
        EADMM    &  4        & Reachable  & 222.3    &  225.0          &      298.0            & 119.0 & 1.65 & 1.68 & 2.94 & 0.87\\
        EADMM    &  6           & Reachable  & 178.5    &  180.0          &   206.0            & 131.0 & 1.34 & 1.34 & 2.24 & 0.97\\
        EADMM    &  10           & Reachable  & 228.7   &  231.0          &   280.0            & 173.0 & 1.66 & 1.66 & 2.57 & 1.24\\
        ADMM     & 4        & Unreachable  & 90.5         & 80.0           & 721.0             & 68.0 & 3.24 & 2.87 & 25.90 & 2.41\\
        ADMM     & 6        & Unreachable  &  112.5        & 110.0            & 483.0	       & 91.0  & 4.03 & 3.91 & 17.45 & 3.23\\
        ADMM     & 10        & Unreachable  &  159.3        & 160.0            & 319.0  	   & 126.0 & 5.61 & 5.64 & 12.56 & 4.43\\
        EADMM    &  2        & Unreachable  & 717.6   &   662.0         &  1418.0              & 526.0 & 5.31 & 4.90 & 10.43 & 3.84\\
        EADMM    &  6           & Unreachable  & 235.3  &  221.0          &  459.0             & 193.0 & 1.85 & 1.74 & 3.69 & 1.44\\
        EADMM    &  10           & Unreachable  & 289.0    &  295.0          &  343.0          & 227.0 & 2.08 & 2.12 & 3.13 & 1.65\\
        \hline
    \end{tabular}\\
\caption{\small Performance of Algorithm~\ref{alg:ADMM_for_MPCT} (ADMM) and EADMM \cite{krupa2021implementation} applied to the ball and plate system for random current states.}
\label{table:iterations_comparisson}
\end{table*}}

Using version \texttt{v0.3.11} of the Spcies Toolbox \cite{Spcies} in an I5-1135G7 laptop, Table~\ref{table:iterations_comparisson} shows results on the number of iterations and computation time of Algorithm~\ref{alg:ADMM_for_MPCT} and the EADMM solver from \cite{krupa2021implementation} (for different values of $\rho$) when applied to the above system for $500$ random initial states, where the position and velocity of the ball are taken from a uniform distribution in the intervals $[0.3, 1.8]$ and $[-0.2, 0.2]$, respectively.
Additionally, we consider two references: the \textit{Reachable} reference $x_r = (1, 0, 0, 0, 0.8, 0, 0, 0)$, and the \textit{Unreachable} reference $x_r = (2.15, 0, 0, 0, 2.2, 0, 0, 0)$, which violates the constraints on the position of the ball.
In both cases $u_r = (0, 0)$.

The results indicate that Algorithm~\ref{alg:ADMM_for_MPCT} is noticeably faster than the EADMM solver when the reference is \textit{Reachable}. However, when it is \textit{Unreachable}, the number of iterations and computation times grow considerably for Algorithm~\ref{alg:ADMM_for_MPCT}. 
Even though ADMM seems less efficient in that specific case, asymptotic convergence of the algorithm is guaranteed for any positive value of $\rho$ \cite{MAL-016}, while EADMM applied to our case study requires $\rho \in (0,0.0176]$ (see \cite[Theorem~3.1]{cai2017convergence}). If we take $\rho$ for EADMM in the range such that its asymptotic convergence is guaranteed, Algorithm~\ref{alg:ADMM_for_MPCT} outperforms EADMM in both the \textit{Reachable} and \textit{Unreachable} cases, as it leads to a poor performance of EADMM.

\begin{figure*}[t]
    \centering
    \begin{subfigure}{0.47\textwidth}
        \includegraphics[scale = 0.45]{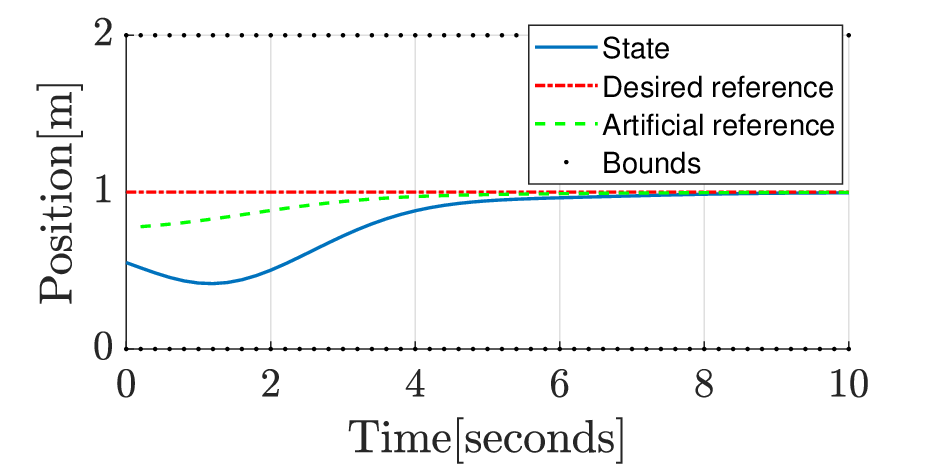}
        \caption{Position in the horizontal axis.}
    \end{subfigure}%
    \hfill%%
    \begin{subfigure}{0.47\textwidth}
        \includegraphics[scale = 0.45]{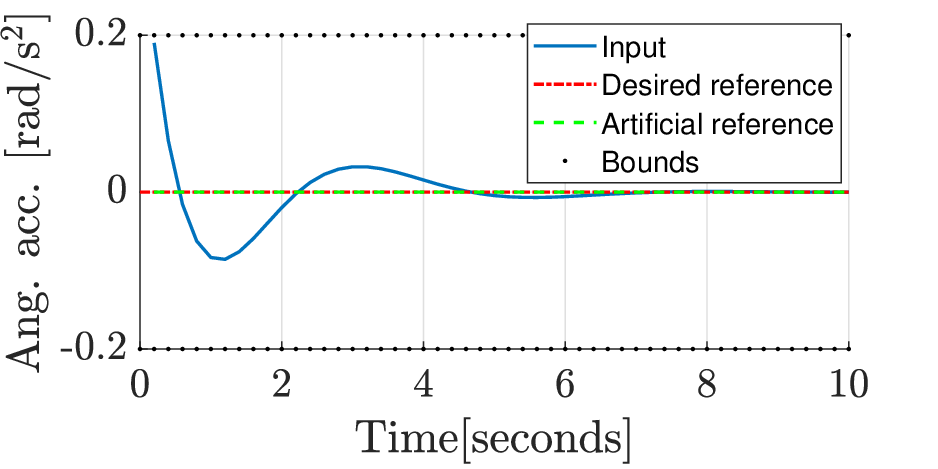}
        \caption{Angular acceleration of the horizontal axis.}
    \end{subfigure}%
    \caption{Closed-loop trajectory of the ball and plate system with \eqref{MPCT_formulation} for an admissible reference.}
    \label{fig:feasible experiment}
\end{figure*}

\begin{figure*}[t]
    \centering
    \begin{subfigure}{0.47\textwidth}
        \includegraphics[scale = 0.45]{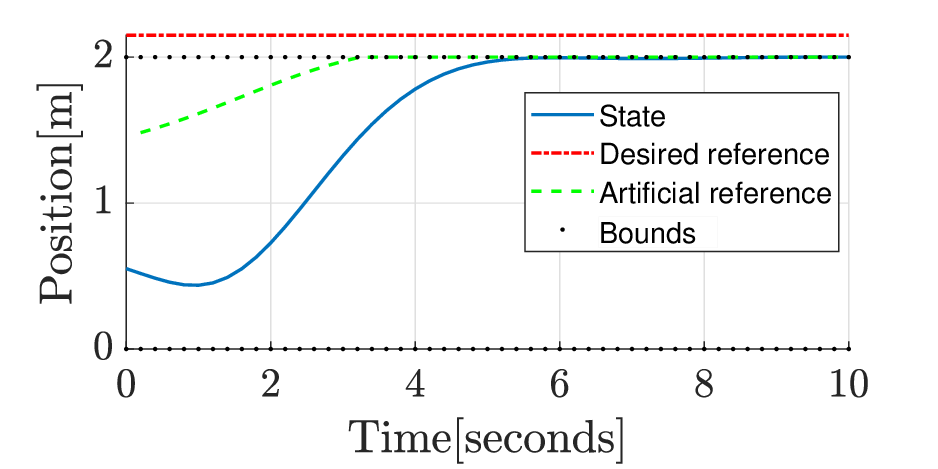}
        \caption{Position in the horizontal axis.}
    \end{subfigure}%
    \hfill%%
    \begin{subfigure}{0.47\textwidth}
        \includegraphics[scale = 0.45]{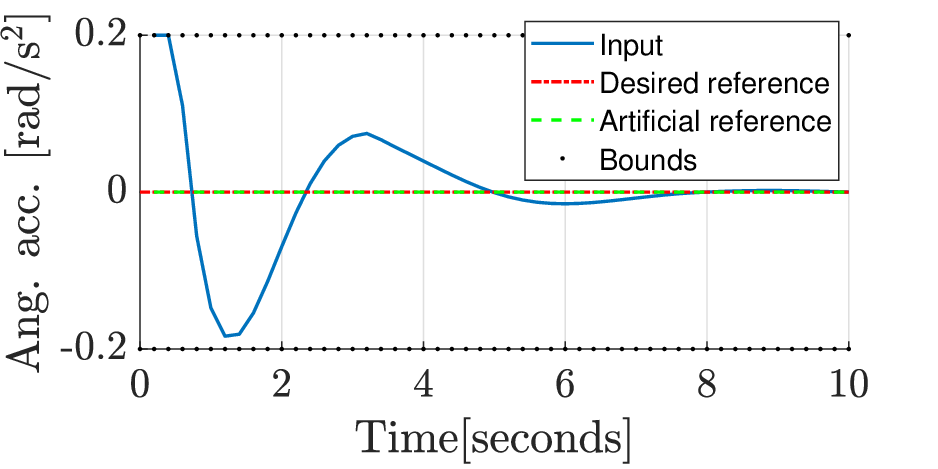}
        \caption{Angular acceleration of the horizontal axis.}
    \end{subfigure}%
    \caption{Closed-loop trajectory of the ball and plate system with \eqref{MPCT_formulation} for a non-admissible reference.}
    \label{fig:unfeasible experiment}
\end{figure*}

Figure~\ref{fig:feasible experiment} illustrates the closed-loop evolution of the system controlled by applying \eqref{MPCT_formulation} to one of the $500$ tests with \textit{Reachable} reference from Table~\ref{table:iterations_comparisson}.
Figure~\ref{fig:unfeasible experiment} shows the closed-loop evolution considering instead the \textit{Unreachable} reference.
As mentioned in Section~\ref{sec:mpct_formulation}, Figure~\ref{fig:unfeasible experiment} shows that when the reference is infeasible, the MPCT formulation \eqref{MPCT_formulation} steers the system to its closest admissible steady-state.

\section{Conclusion} \label{sec:conclusion}

This article has presented an efficient method for solving the MPCT formulation using ADMM.
We have shown how the semi-banded structure of the most computationally-expensive step of the ADMM algorithm can be solved efficiently by decomposing it into three simpler-to-solve steps.
Indeed, the decomposition recovers the simple matrix structure exploited by several first-order solvers for standard MPC from the literature.
We have presented numerical results comparing the proposed solver with a recent EADMM-based MPCT solver from the literature, showing that the proposed approach may outperform the EADMM algorithm from a computational point of view.
This, along with the solid theoretical convergence guarantees of ADMM, results in a rather sensible and efficient solution for the practical implementation of MPCT. 

\bibliographystyle{IEEEtran}
\bibliography{IEEEabrv,semibanded_MPCT_ECC}
	
\end{document}